%% file: main.tex
\documentclass[runningheads,a4paper]{llncs}

\usepackage{times}
\usepackage{amssymb,latexsym,amsmath}     
\usepackage{thmtools, thm-restate}

\usepackage{amsfonts}
\usepackage{graphicx}
\usepackage{enumerate}
\usepackage[usenames,dvipsnames]{xcolor}
\usepackage{framed}
\usepackage{tikz}
\usepackage{float}
\usepackage{textcomp}
\usepackage{url}
\urldef{\mailsa}\path|{hajiagha, hmahini, asawant}@cs.umd.edu|

\begin{document}

\mainmatter  

\title{Scheduling a Cascade with Opposing Influences}

\author{MohammadTaghi Hajiaghayi\and Hamid Mahini\and Anshul Sawant}

\institute{University of Maryland at College Park
\\ \mailsa 
}
\maketitle

\begin{abstract}
Adoption or rejection of ideas, products, and technologies in a society is often governed by simultaneous propagation of positive and negative influences. Consider a planner trying to introduce an idea in different parts of a society at different times. How should the planner design a schedule considering this fact that positive reaction to the idea in early areas has a positive impact on probability of success in later areas, whereas a flopped reaction has exactly the opposite impact? 
We generalize a well-known economic model which has been recently used by Chierichetti, Kleinberg, and Panconesi (ACM EC'12). In this model the reaction of each area is determined by its initial preference and the reaction of early areas.
We model the society by a graph where each node represents a group of people with the same preferences. We consider a full propagation setting where news and influences propagate between every two areas. 
We generalize previous works by studying the problem when people in different areas have various behaviors.

We first prove, independent of the planner's schedule, influences help (resp., hurt) the planner to propagate her idea if it is an appealing (resp., unappealing) idea. 
We also study the problem of designing the optimal non-adaptive spreading strategy. In the non-adaptive spreading strategy, the schedule is fixed at the beginning and is never changed. Whereas, in adaptive spreading strategy the planner decides about the next move based on the current state of the cascade.   We demonstrate that it is hard to propose a non-adaptive spreading strategy in general. Nevertheless, we propose an algorithm to find the best non-adaptive spreading strategy when probabilities of different behaviors of people in various areas drawn i.i.d from an unknown distribution. 
Then, we consider the influence propagation phenomenon when the underlying influence network can be any arbitrary graph. We show it is $\#P$-complete to compute the expected number of adopters for a given spreading strategy. However, we design a polynomial-time algorithm for the problem of computing the expected number of adopters for a given schedule in the full propagation setting. 
Last but not least, we give a polynomial-time algorithm for designing an optimal adaptive spreading strategy in the full propagation setting.

\end{abstract}

\input{intro}

\input{notation}

\input{nonadaptive}

\section*{Acknowledgments}
\label{sec:Acknowledgments}
Authors would like to thank Jon Kleinberg for his useful comments about the motivation of our problem.


\appendix
\input{appendix.tex}

\end{document}

%% file: intro.tex
\def\yes{\mathcal{Y}}
\def\no{\mathcal{N}}

\def\cat{area}
\def\cats{areas}
\def\Cat{Area}
\def\Cats{Areas}
\def\iid{i.i.d.}
\def\strategy{spreading strategy}
\def\strategies{spreading strategies}
\def\Strategy{Spreading strategy}

\section{Introduction}
People's opinions are usually formed by their friends' opinions. 
Whenever a new concept is introduced into a society, the high correlation between people's reactions initiates an influence propagation. 
Under this propagation, the problem of promoting a product or an opinion depends on the problem of directing the flow of influences. 
As a result, a planner can develop a new idea by controlling the flow of influences in a desired way. Although there have been many attempts to understand the behavior of influence propagation in a social network, the topic is still controversial due to lack of reliable information and complex behavior of this phenomenon. 
For example, one compelling approach is ``seeding'' which was introduced by the seminal work of Kempe, Kleinberg, and Trados \cite{KKT03} and is well-studied in the literature \cite{KKT03,KKT05,MR07}. The idea is to influence a group of people in the initial investment period and spread the desired opinion in the ultimate exploitation phase. 
Another approach is to use time-varying and customer-specific prices to propagate the product (see e.g., \cite{AEGHIMM10,AGHMMN10,HMS08}).
All of these papers investigate the influence propagation problem when only positive influences spread into the network.
However, in many real world applications people are affected by both positive and negative influences, e.g., when both consenting and dissenting opinions broadcast simultaneously.

We generalize a well-known economic model introduced by Arthur \cite{Arthur89}. This model has been recently used by Chierichetti, Kleinberg, and Panconesi \cite{CKP12}. 
Assume an organization is going to develop a new idea in a society where the people in the society are grouped into $n$ different areas. 
Each area consists of people living near each other with almost the same preferences.
The planner schedules to introduce a new idea in different areas at different times. Each area may accept or reject the original idea. Since areas are varied and effects of early decisions boost during the diffusion, a schedule-based strategy affects the spread of influences. 
This framework closely matches to various applications from economics to social science to public health where the original idea could be a new product, a new technology, or a new belief. 
Consider the spread of two opposing influences simultaneously. Both positive and adverse reactions to a single idea originate different flows of influences simultaneously. In this model, each \cat\ has an {\em initial preference} of $\mathcal{Y}$ or $\mathcal{N}$. The initial preference of $\mathcal{Y}$ ($\mathcal{N}$) means  the area will accept (decline) the original idea when there are no network externalities. Let $c_i$ be a non-negative number indicating how reaction of  people in area  $i$ depends on the others'. 
We call $c_i$ the {\em threshold} of area $i$. Assume the planner introduced the idea in area $i$ at time $s$. Let $m_{\yes}$ and $m_{\no}$ be the number of areas which accept or reject the idea before time $s$. If $|m_{\yes} - m_{\no}| \geq c_i$ the people in area $i$ decide based on the majority of previous adopters. It means they adopt the idea if $m_{\yes} - m_{\no} \geq c_i$ and drop it if $m_{\no} - m_{\yes} \geq c_i$. 
Otherwise, if $|m_{\yes} - m_{\no}| < c_i$ the people in area $i$ accept or reject the idea if the initial preference of area $i$ is  $\mathcal{Y}$ or $\mathcal{N}$ respectively. 
The planner does not know exact initial preferences and has  only prior knowledge about them. Formally speaking, for area $i$ the planner knows the initial preference of area $i$ will be $\mathcal{Y}$ with probability $p_i$ and will be $\mathcal{N}$ with probability $1-p_i$. We call $p_i$ the {\em initial acceptance probability} of area $i$.

We consider the problem when the planner classifies different areas into various types. The classification is based on the planner's knowledge about the reaction of people living in each area. Hence, the classification is based on different features, e.g., preferences, beliefs, education, and age such that people in areas with the same type react almost the same to the new idea. It means all areas of the same type have the same threshold $c_i$ and the same initial acceptance probability $p_i$. It is worth mentioning previous works only consider the problem when all areas have the same type, i.e., all $p_i$'s and $c_i$'s are the same \cite{Arthur89,CKP12}.
The planner wants to manage the flow of influences, and her {\em \strategy\ } is a permutation $\pi$ over different areas. 
Her goal is to find a spreading strategy $\pi$ which maximizes the expected number of adopters. 
We consider both {\em adaptive} and {\em non-adaptive} spreading strategies in this paper. In the adaptive spreading strategy, the planner can see results of earlier areas for further decisions. On the other hand, in the non-adaptive spreading strategy the planner decides about the permutation in advance.
We show the effect of a spreading strategy on the number of adopters with an example in Appendix \ref{sec:example}.

\input{relatedwork.tex}

\input{results.tex}

%% file: relatedwork.tex
\subsection{Related Work}
We are motivated by a series of well-known studies in economics and politics literature in order to model people's behavior \cite{Arthur89,B92,BHW92,G78}. Arthur first proposed a framework to analyze people's behavior in a scenario with two competing products \cite{Arthur89}. In this model people are going to decide about one of two competing products alternatively.  He studied the problem when people are affected by all previous customers, and the planner has the same prior knowledge about people's behavior, i.e., people have the same types. He demonstrated that a cascade of influences is formed when products have positive network externalities, and early decisions determine the ultimate outcome of the market. It has been showed the same cascade arises when people look at earlier decisions, not because of network externalities, but because they have limited information themselves or even have bounded rationality to process all available data \cite{B92,BHW92}. 

Chierichetti, Kleinberg, and Panconesi argued when relations between people form an arbitrary network, the outcome of an influence propagation highly depends on the order in which people make their decisions \cite{CKP12}. In this setting, a potential spreading strategy is an ordering of decision makers. They studied the problem of finding a spreading strategy which maximizes the expected number of adopters when people have the same type, i.e., people have the same threshold $c$ and the same initial acceptance probability $p$. They proved for any $n$-node graph there is an adaptive spreading strategy with at least $O(np^c)$ adopters. They also showed for any $n$-node graph all non-adaptive spreading strategies result in at least (resp. at most) $\frac{n}{2}$ if initial acceptance probability is less (resp. greater) than $\frac{1}{2}$. 
They considered the problem on an arbitrary graph when nodes have the same type. While we mainly study the problem on a complete graph when nodes have various types, we improve their result in our setting and show the expected number of adopters for all adaptive spreading strategies is at least (resp. at most) $np$ if initial acceptance probability is $p \geq \frac{1}{2}$ (resp. $p \leq \frac{1}{2}$).
We also show the problem of designing the best spreading strategy is hard on an arbitrary graph with several types of customers. 
We prove it is $\#P$-complete to compute the expected number of adopters for a given spreading strategy.

The problem of designing an appropriate marketing strategy based on network externalities has been studied extensively in the computer science literature. For example, Kempe, Kleinberg, and Tardos \cite{KKT03} studied the following question in their seminal work: How can we influence a group of people in an investment phase in order to propagate an idea in the exploitation phase? This question was introduced by Domingos and Richardson \cite{DR}. The answer to this question leads to a marketing strategy based on seeding. There are several papers that study the same problem from an algorithmic point of view, e.g., \cite{KKT05,MR07,chen2009efficient}.
Hartline, Mirrokni, and Sundararajan \cite{HMS08} also proposed another marketing strategy based on scheduling for selling a product. Their marketing strategy is a permutation $\pi$ over customers and  price $p_i$ for customer $i$. The seller offers the product with price $p_i$ to customer $i$ at time $t$ where $t = \pi^{-1}(i)$. The goal is to find a marketing strategy which maximizes the profit of the seller. This approach is followed by several works, e.g., \cite{AEGHIMM10,AGHMMN10,AMSX09}. These papers study the behavior of an influence propagation when there is only one flow on influences in the network. In this paper, we study the problem of designing a spreading strategy when both negative and positive influences propagate simultaneously.

The propagation of competitive influences has been studied in the literature (See \cite{GK12} and its references). These works studied the influence propagation problem in the presence of competing influences, i.e., when two or more competing firms try to propagate their products at the same time. However we study the problem of influence propagation when there exist both positive and negative reactions to the same idea. There are also studies which consider the influence propagation problem in the presence of positive and negative influences \cite{new11,new13}. Che et al. \cite{new11} use a variant of the independent cascade model introduced in \cite{KKT03}. They model negative influences by allowing each person to flips her idea with a given probability $q$. Li et al. \cite{new13} model the negative influences by negative edges in the graph. Although they study the same problem, we use different models to capture behavior of people.

%% file: results.tex
\subsection{Our Results}
We analyze an influence propagation phenomenon where two opposing flows of influences propagate through a social network. As a result, a mistake in the selection of early areas may result in propagation of negative influences. 
Therefore a good understanding of influence propagation dynamics seems necessary to analyze the properties of a spreading strategy. Besides the previous papers which have studied the problem with just one type \cite{Arthur89,CKP12}, we consider the scheduling problem with various types. 
%
%
Also, we mainly study the problem in a {\em full propagation} setting as it matches well to our motivations. In the full propagation setting news and influences propagate between every two areas.  
One can imagine how internet, media, and electronic devices broadcast news and influences from everywhere to everywhere. 
In the {\em partial propagation} setting news and influences do not necessarily propagate between every two areas. In the partial propagation setting the society can be modeled with a graph, where there is an edge from area $i$ to area $j$ if and only if influences propagate from area $i$ to area $j$. 

Our main focus is to analyze the problem when the planner chooses a non-adaptive spreading strategy.
Consider an arbitrary non-adaptive spreading strategy when initial preferences of all \cats\ are $p$. 
The expected number of adopters is exactly $np$ if all areas decide independently.
We demonstrate that in the presence of network influences, the expected number of adopters is greater/less than $np$ if initial acceptance probability $p$ is greater/less than $\frac{1}{2}$.
These results have a bold message: {\bf The influence propagation is an amplifier for an appealing idea and an attenuator for an unappealing idea.}
%
Chierichetti, Kleinberg, and Panconesi \cite{CKP12} studied the problem on an arbitrary graph with only one type. They proved the number of adopters is greater/less than $\frac{n}{2}$ if initial acceptance probability $p$ is greater/less than $\frac{1}{2}$. Theorem \ref{sec:bounds-spread-weak} improves their result from $\frac{n}{2}$ to $np$ in our setting. 
All missing proofs are in the full version of
the paper.
\begin{restatable}{mythm}{np}
\label{sec:bounds-spread-weak}
Consider an arbitrary non-adaptive spreading strategy $\pi$ in the full propagation setting. Assume all initial acceptance probabilities are equal to $p$. If $p \geq \frac{1}{2}$, then the expected number of adopters is at least $np$. Furthermore, If $p \leq \frac{1}{2}$, then the expected number of adopters is at most $np$.
\end{restatable}

Chierichetti, Kleinberg, and Panconesi \cite{CKP12} studied the problem of designing an optimum spreading strategy in the partial propagation setting. They design an approximation algorithm for the problem when the planner has the same  prior knowledge about all areas, i.e., all areas have the same type. We study the same problem with more than one type.
%
We first consider the problem in the full propagation setting.
%
%
One approach is to consider a non-adaptive spreading strategy with a constant number of switches between different types. 
The planner has the same prior knowledge about areas with the same type. It means areas with the same type are identical for the planner. Thus any spreading strategy can be specified by types of areas rather than areas themselves. Let $\tau(i)$ be the type of area $i$ and $\tau(\pi)$ be the sequence of types for spreading strategy $\pi$. For a given spreading strategy $\pi$ a {\em switch} is a position $k$ in the sequence such that $\tau(\pi(k)) \neq \tau(\pi(k+1))$. 
As an example consider a society with $4$ areas. Areas $1$ and $2$ are of type $1$. Areas $3$ and $4$ are of type $2$. Then spreading strategy $\pi_1=(1, 2, 3, 4)$ with $\tau(\pi_1)=(1, 1, 2, 2)$ has a switch at position $2$ and spreading strategy $\pi_2=(1, 3, 2, 4)$ with $\tau(\pi_2)=(1, 2, 1, 2)$ has switches at positions $1$, $2$, and $3$.
\begin{restatable}{mythm}{switch}
\label{thm:switch}
A $\sigma$-switch spreading strategy is a spreading strategy with at most $\sigma$ switches.
For any constant $\sigma$, there exists a society with areas of two types such that no $\sigma$-switch spreading strategy is optimal.
\end{restatable}
We construct a society with $n$ areas with $\frac{n}{2}$ areas of type $1$ and $\frac{n}{2}$ areas of type $2$. We demonstrate an optimal non-adaptive spreading strategy should switch at least $\Omega(n)$ times. It means no switch-based non-adaptive spreading strategy can be optimal. 
We prove Theorem \ref{thm:switch} formally in Appendix \ref{sec:switch}.

%

On the positive side, we analyze the problem when thresholds are drawn independently from an unknown distribution and initial acceptance probabilities are arbitrary numbers. We characterize the optimal non-adaptive spreading strategy in this case.
\begin{restatable}{mythm}{iidc}
\label{thm-iidc}
Assume that the planner's prior knowledge about all values of $c_i$'s is the same, i.e., all $c_i$'s are drawn independently from the same but unknown distribution. Let initial acceptance probabilities be arbitrary numbers.  Then, the best non-adaptive spreading strategy is to order all areas in non-increasing order of their initial acceptance probabilities.
\end{restatable}

%
We also study the problem of designing the optimum spreading strategy in the partial propagation setting with more than one types. We show it is hard to determine the expected number of adopters for a given spreading strategy.
Formally speaking, we show it is $\#P$-complete to compute the expected number of adopters for a given spreading strategy $\pi$ in the partial propagation setting with more than one type. 
This is another evidence to show the influence propagation is more complicated with more than one type. 
We prove Theorem \ref{thm:hardness} based on a reduction from a variation of the {\em network reliability} problem in Appendix \ref{sec:hardness}.

\begin{restatable}{mythm}{hardness}
\label{thm:hardness}
In the partial propagation setting, it is $\#P$-complete to compute the expected number of adopters for a given non-adaptive spreading strategy $\pi$.
\end{restatable}

We also present a polynomial-time algorithm to compute the expected number of adopters for a given non-adaptive spreading strategy in a full propagation setting.
We design an algorithm in order to simulate the amount of propagation for a given spreading strategy in Appendix \ref{sec:comp-expect-numb}. 
\begin{restatable}{mythm}{givenschedule}
\label{thm:compute_dp}
Consider a full propagation setting. The expected number of adopter can be computed in polynomial time for a given non-adaptive spreading strategy $\pi$.  
\end{restatable}

At last we study the problem of designing the best adaptive spreading strategy. We overcome the hardness of the problem and design a polynomial-time algorithm to find the best adaptive marketing strategy in the following theorem. We describe the algorithm precisely in Appendix \ref{sec:best_schedule_dp}.
\begin{restatable}{mythm}{adaptive_dp}
\label{thm:best_schedule_dp}
A polynomial-time algorithm finds the best adaptive spreading strategy for a society with a constant number of types. 
\end{restatable}

%% file: notation.tex
\newcommand{\fa}{~\forall}
\newcounter{thm}
\newcounter{c}
\newtheorem{lem}[thm]{Lemma}	   
\newtheorem{fact}[thm]{Fact}	   
\newcounter{defc}		   
\newtheorem{defn}[defc]{Definition}
\newcommand{\nop}[1]{{}}
\newcommand{\lref}[1]{Lemma \ref{#1}}
\newcommand{\induction}[3]{\par\text{\textsl{Induction Hypothesis:} }#1
\par\text{\textsl{Basis:} }#2
\par\text{\textsl{Induction Step:} }#3}

\section{Notation and Preliminaries}
In this section we define basic concepts and notation used throughout this paper. We first formally define the spread of influence through a network as a stochastic process and then give the intuition behind the formal notation.
We are given a graph $G=(V,E)$ with thresholds, $c_v \in \mathbb{Z}_{>0},\forall v \in V$ and initial acceptance probabilities $p_v \in [0,1],\forall v \in V$. Let $|V| = n$. Let $d_v$ be the degree of vertex $v$. Let $N(v)$ be the set of neighboring vertices of $v$. Let $\boldsymbol{c}$ be the vector $(c_1,\ldots,c_n)$ and $\boldsymbol{p}$ be the vector $(p_1, \ldots, p_n)$. Given a graph $G=(V,E)$ and a permutation $\pi:V \mapsto V$, we define a discrete stochastic process, \emph{IS} (Influence Spread) as an ordered set of random variables $(X^1, X^2, \ldots, X^n)$, where $X^t \in \Omega = \{-1,0,1\}^n, \forall t \in \{1,\ldots,n\}$. The random variable $X^t_v$ denotes decision of area $v$ at time $t$. If it has not yet been scheduled, $X^t_v=0$. If it accepts the idea then $X^t_v=1$, and if it rejects the idea then $X^t_v=-1$. Note that $X^t_v=0$ iff $t < \pi^{-1}(v)$. Let $D(v) = \sum_{u \in N(v)}X^{\pi^{-1}(v)}_u$ be the sum of decision's of $v$'s neighbors. For simplicity in notation, we denote $X_v^n$ by $X_v$.

We now briefly explain the intuition behind the notation. The input graph models the influence network of areas on which we want to schedule a cascade, with each vertex representing an area. There is an edge between two vertices if two corresponding areas influence each others decision. The influence spread process models the spread of idea acceptance and rejection for a given spreading strategy. The permutation $\pi$ maps a position in spreading strategy to an area in $V$. For example, $\pi(1)=v$ implies that $v$ is the first area to be scheduled. Once the area $v$ is given a chance to accept or reject the idea at time $\pi^{-1}(v)$, $X^{\pi^{-1}(v)}_v$ is assigned a value based on $v$'s decision and at all times $t$ after $\pi^{-1}(v)$, $X^t_v = X^{\pi^{-1}(v)}_v$. The random variable $X_v$ denotes whether an area $v$ accepted or rejected the idea. We note that $X^{t}_v = X_v, \forall t \geq \pi^{-1}(v)$. The random variable $X^t$ is complete snapshot of the cascade process at time $t$. The variable $D(v)$ is the decision variable for $v$. It denotes the sum of decisions of $v$'s neighbors at the time $v$ is scheduled in the cascade and it determines whether $v$ decides to follow the majority decision or whether $v$ decides based on its initial acceptance probability. The random variable $I_t$ is the sum of decisions of all areas at time $t$. Thus, $I_n$ is the variable we are interested in as it denotes the difference between number of people who accept the idea and people who reject the idea.

Let $v = \pi(t)$. Given $X^{t-1}$, $X^t$ is defined as follows:
\begin{itemize}
\item  Every area decides to accept or reject the idea exactly once when it is scheduled and its decision remains the same at all later times. Therefore $\forall i \ne \pi(t)$:
  \begin{itemize}
  \item   $X^t_i = X^{t-1}_i$
  \end{itemize}
\item  Decision of area $v$ is based on decision of previous areas if its threshold is reached.
  \begin{itemize}
  \item $X^{t}_{v}  = 1 \text{ if }D(v) \geq c_v$
  \item $X^{t}_{v}  = -1  \text{ if }D(v) \leq -c_v$
  \end{itemize}
\item If threshold of area $v$ is not reached, then it decides to accept the idea with probability $p_v$, its initial acceptance probability, and decides to reject it with
  probability $1-p_v$.
\end{itemize}
%

In partial propagation setting, we represent such a stochastic process by tuple $IS=(G, \boldsymbol{c}, \boldsymbol{p}, \pi)$. For full propagation setting, the underlying graph is a complete graph and hence we can denote the process by $(\boldsymbol{c}, \boldsymbol{p}, \pi)$. When $\boldsymbol{c}$ and $\boldsymbol{p}$ are clear from context, we denote the process simply by spreading strategy, $\pi$. We define random variable $I_t = \sum_{v \in V}X^t_v$. We denote by $q_v = 1 - p_v$ the probability that $v$ rejects the idea based on initial preference. We denote by $Pr(A; IS)$, the probability of event $A$ occurring under stochastic process $IS$. Similarly, we denote by $E(z; IS)$, the expected value of random variable $z$ under the stochastic process $IS$.

%% file: nonadaptive.tex
\input{np}

\input{randomth}

%% file: np.tex
\section{A Bound on Spread of Appealing and Unappealing Ideas}
\label{sec:np}
Lets call an idea unappealing if its initial acceptance probability for all areas is $p$ for some $p \leq \frac{1}{2}$. We prove in this section, that for such ideas, no strategy can boost the acceptance probability for any area above $p$. We note that exactly the opposite argument can be made when $p \geq \frac{1}{2}$ is the initial acceptance probability of all areas, i.e., any spreading strategy guarantees that every area accepts the idea with probability of at least $p$. 
\np*
\begin{proof}
  We prove this result for the case when $p \leq \frac{1}{2}$. The other case ($p \leq \frac{1}{2}$) follows from symmetry.
  To avoid confusion, we let $p_0 = p$ and use $p_0$ instead of the real number $p$ throughout this proof.
  If we prove that any given area accepts the idea with probability of at most $p_0$, then from linearity of expectation, we are done. Consider an area $v$ scheduled at time $t+1$. The probability that the area accepts or rejects the idea is given by
  \begin{align*}
    Pr(X_v = 1) =& p_0(1 - Pr(I_t \geq c_v) - Pr(I_t \leq -c_v)) + Pr(I_t \geq c_v),\\
    Pr(X_v = -1) =& (1-p_0)(1 - Pr(I_t \geq c_v) - Pr(I_t \leq -c_v)) + Pr(I_t \leq -c_v).
  \end{align*}
  Since $Pr(X_v=1)+Pr(X_v=-1) = 1$, if we prove that $\frac{Pr(X_v = 1)}{Pr(X_v = -1)} \leq \frac{p_0}{1-p_0}$, then we have $Pr(X_v = 1) \leq p_0$.
  We have
  \begin{align*}
    \frac{Pr(X_v = 1) }{Pr(X_v = -1)}=\frac{ p_0(1 - Pr(I_t \geq c_v) - Pr(I_t \leq -c_v)) + Pr(I_t \geq c_v)}
    {(1-p_0)(1 - Pr(I_t \geq c_v) - Pr(I_t \leq -c_v)) + Pr(I_t \leq -c_v)}.
  \end{align*}
  We have:
  \begin{align*}
    \frac{p_0(1 - Pr(I_t \geq c_v) - Pr(I_t \leq -c_v))}{(1-p_0)(1 - Pr(I_t \geq c_v) - Pr(I_t \leq -c_v))} = \frac{p_0}{1-p_0}.
  \end{align*}  
  We know that for any $a, b, c, d, e  \in \mathbb{R}_{>0}$, if $\frac{a}{b} \leq e$ 
  and $\frac{c}{d} \leq e$ then 
  \begin{align}
    \frac{a+c}{b+d} \leq e.\label{eq:17}
  \end{align}
  Therefore, if we prove that $\frac{Pr(I_t \geq c_v)}{Pr(I_t \leq -c_v)} \leq \frac{p_0}{1-p_0}$,
  we are done. Thus, we can prove this theorem by proving that $\frac{Pr(I_k \geq x)}{Pr(I_k \leq -x)} \leq \frac{p_0}{1-p_0}$ for all $x \in \{1 \ldots k\}, k \in \{1\ldots n\}$. We prove this by induction on number of areas.
If there is just one area, then that area decides to accept with probability $p_0$ (as all initial acceptance probabilities are equal to $p_0$).
Assume if the number of areas is less than or equal to $n$, then $\frac{Pr(I_k \geq x)}{Pr(I_k \leq -x)} \leq \frac{p_0}{1-p_0}$ for all $x \in \{1 \ldots k\}, k \in \{1\ldots n\}$. We prove the statement when there are $n+1$ areas.
Let $par(n,x):\mathbb{N}\times\mathbb{N}\mapsto\{0,1\}$ be a function which is $0$ if $n$ and $x$ have the same parity,
    $1$ otherwise.
    Let $v$ be the area scheduled at time $n+1$. Let $\nu = par(n, x)$.
    We now consider the following three cases.\\
    \textbf{Case 1:} $1 \leq x \leq n-2$. The event $I_{n+1} \geq x+1$ is the union of the following two disjoint events:
    \begin{enumerate}
    \item $I_{n} \geq x+2$, and whatever the $n^{th}$ area decides, $I_{n+1}$ is at least $x+1$.
    \item $I_{n} = x+\nu$ and $n+1^{th}$ area decides to accept.
    \end{enumerate}
    Similarly, the event $I_{n+1} \leq -x-1$ is the union of the event $I_n \leq -x-2$ and the event --- $I_n = -x - \nu$ and the $n+1^{th}$ area rejects the idea.
    We note that we require the $par$ function because only one of the events $I_{n}=x$ and $I_n = x+1$ can occur w.p.p. depending on parities of $n$ and $x$. 
    Thus
    \begin{align*}
      Pr(I_{n+1} \geq x+1) =& Pr(I_n \geq x+2) + Pr(X_v=1|I_n = x + \nu)Pr(I_n=x+\nu),\\
      Pr(I_{n+1} \leq -x-1) =& Pr(I_n \leq -x-2) + Pr(X_v=-1|I_n = -x -\nu)Pr(I_n=-x-\nu).
    \end{align*}
    Now, if $x + \nu \geq c_v$, then $Pr(X_v=1|I_n = x + \nu) = Pr(X_v=-1|I_n = -x -\nu) = 1$, otherwise $Pr(X_v=1|I_n = x + \nu) = p_0 < 1-p_0 = Pr(X_v=-1|I_n = -x -\nu)$.
    Therefore, $Pr(X_v=1|I_n = x + \nu) \leq Pr(X_v=-1|I_n = -x -\nu)$. Let $\beta = Pr(X_v=-1|I_n = -x -\nu)$.
    Using the above, we have
    \begin{align*}
      Pr(I_{n+1} \geq x+1) \leq& Pr(I_n \geq x+2) + \beta Pr(I_n=x+\nu),\\
      Pr(I_{n+1} \leq -x-1) =& Pr(I_n \leq -x-2) + \beta Pr(I_n=-x-\nu).
    \end{align*}
    From above, we have
    \begin{align}
      f(\beta) =& \frac{Pr(I_n \geq x+2) + \beta Pr(I_n=x+\nu)}{Pr(I_n \leq -x-2) + \beta Pr(I_n=-x-\nu)}
      \geq \frac{Pr(I_{n+1} \geq x+1)}{Pr(I_{n+1} \geq -x-1)}. \label{eq:12}
    \end{align}
    The function $f(\beta)$ is either increasing or decreasing and hence has extrema at end points of its range.
    The maxima is {$\leq \max\{\frac{Pr(I_n \geq x+2)}{Pr(I_n \leq -x-2)}, \frac{Pr(I_n \geq x+2) + Pr(I_n=x+\nu)}{Pr(I_n \leq -x-2) + Pr(I_n=-x-\nu)}\}$} because $\beta \in [0,1]$. 
    Now $Pr(I_n \geq x+2) + Pr(I_n=x+1) + Pr(I_n=x) = Pr(I_n \geq x)$ and $Pr(I_n \leq -x-2) + Pr(I_n=-x-\nu) = Pr(I_n \leq -x)$.
    Thus $f \leq \max\{\frac{Pr(I_n \geq x+2)}{Pr(I_n \leq -x-2)}, \frac{Pr(I_n \geq x)}{Pr(I_n \leq -x)}\} \leq \frac{p_0}{1-p_0}$ (from induction hypothesis).
    From above and (\ref{eq:12}), $\frac{Pr(I_{n+1} \geq x+1)}{Pr(I_{n+1} \leq -x-1)} \leq \frac{p_0}{1-p_0}$.\\
    \textbf{Case 2:} $x=0$. If $n$ is odd then $Pr(I_{n+1} \geq 1) = Pr(I_{n+1} \geq 2)$ and $Pr(I_{n+1} \leq -1) = Pr(I_{n+1} \leq -2)$ and this case is the same as $x=1$ and hence considered above. Thus, assume that $n$ is even. Thus
    \begin{align}
      Pr(I_{n+1} \geq 1) &= Pr(I_n \geq 2) + Pr(X_v=1|I_n = 0)Pr(I_n=0), \label{eq:13}\\
      Pr(I_{n+1} \leq -1) &= Pr(I_n \leq -2) + Pr(X_v=-1|I_n = 0)Pr(I_n=0). \label{eq:14}
    \end{align}
    Since, if $I_n = 0$, then areas decide based on the initial acceptance probability.
    We have $Pr(X_v=1|I_n = 0) = p_0$ and $Pr(X_v=-1|I_n = 0) = 1-p_0$. Using this fact
    ,by dividing (\ref{eq:13}) and (\ref{eq:14}), we have
    \begin{align*}
      \frac{Pr(I_{n+1} \geq 1) }{Pr(I_{n+1} \leq -1) } \leq \frac{Pr(I_n \geq 2) + p_0Pr(I_n=0)}{Pr(I_n \leq -2) + (1-p_0)Pr(I_n=0)}.
    \end{align*}
    From induction hypothesis, $\frac{Pr(I_{n} \geq 2)}{Pr(I_{n} \leq -2)} \leq \frac{p_0}{1-p_0}$. Thus, we conclude $\frac{Pr(I_{n+1} \geq 1) }{Pr(I_{n+1} \leq -1) } \leq \frac{p_0}{1-p_0}$ based on (\ref{eq:17}).\\
    \textbf{Case 3:} $x \in \{n-1, n\}$.
    In this case $Pr(I_n \geq x+2) =0$, since the number of adopters can never be more than the number of total areas. Also, $I_{n+1} $ cannot be equal to $n$ because $n$ and $n+1$ don't have the same parity. Therefore, $Pr(I_{n+1} \geq n) = Pr(I_{n+1} \geq n+1)$ and $Pr(I_{n+1} \leq -n) = Pr(I_{n+1} \leq -n-1)$. Thus, it is enough to analyze the case $x=n$. We have
    \begin{align*}
      Pr(I_{n+1} \geq n+1) &= Pr(X_v=1|I_n = n)Pr(I_n=n), \\
      Pr(I_{n+1} \leq n+1) &= Pr(X_v=-1|I_n = -n)Pr(I_n=-n). 
    \end{align*}
    Since either both decisions are made based on thresholds with probability $1$ or both are made based on initial probabilities and initial acceptance probability is less than the initial rejection probability, We know that $Pr(X_v=1|I_n = n) \leq Pr(X_v=-1|I_n = -n)$.
    Therefore $\frac{Pr(I_{n+1} \geq n+1)}{Pr(I_{n+1} \leq n+1) } \leq \frac{Pr(I_n=n)}{Pr(I_n = -n)}$.
    Now, since $Pr(I_n = n) = Pr(I_n \geq n)$ and $Pr(I_n = -n) = Pr(I_n \leq -n)$, from induction hypothesis, we have $\frac{Pr(I_{n+1} \geq n+1)}{Pr(I_{n+1} \leq n+1) } \leq \frac{p_0}{1-p_0}$ and we are done.
\end{proof}


%% file: randomth.tex
\section{Non-adaptive Marketing Strategy with Random Thresholds}
\label{sec:random-th}
We consider the problem of designing a non-adaptive spreading strategy when the thresholds are drawn independently from the same but unknown distribution. We show the best spreading strategy is to schedule areas in a non-increasing order of initial acceptance probabilities. We prove the optimality of the algorithm using a coupling argument. 
%
First we state the following lemma which will be useful in proving Theorem \ref{thm-iidc}. The proof is in Appendix \ref{sec:miss:lem1}.
\begin{lem}
  \label{lem-triv}
  Let $\pi$ and $\pi'$ be two spreading strategies. If $\exists k \in \mathbb{Z}_>0$, such that $\pi(i) = \pi'(i), \fa i \geq k$ and $Pr(I_{k} \geq x;\pi) \geq Pr(I_{k} \geq x;\pi'), \fa x \in \mathbb{Z}$, then $E(I_n;\pi) \geq E(I_n;\pi')$.
\end{lem}

\iidc*
\begin{proof}
Let $\pi'$ be a spreading strategy where areas are scheduled in an order that is not non-increasing. Thus, there exists $k$ such that $p_{\pi'(k)} < p_{\pi'(k+1)}$. We prove that if a new spreading strategy $\pi$ is created by exchanging position of areas $\pi'(k)$ and $\pi'(k+1)$, then the expected number of people who accept the idea cannot decrease. It means the best spreading strategy is non-increasing in the initial acceptance probabilites.  

To prove the theorem, we will prove that $Pr(I_{k+1} \geq x; \pi) \geq Pr(I_{k+1} \geq x; \pi')$ and the result then follows from Lemma \ref{lem-triv}.
  Since, the two spreading strategies are identical till time $k-1$ and therefore the random variable $I_{k-1}$ has identical distribution under both the strategies, 
  we can prove the above by proving that 
  $Pr(I_{k+1} \geq I_{k-1} + y|I_{k-1};\pi) \geq Pr(I_{k+1} \geq I_{k-1} + y|I_{k-1};\pi')$ for all $y \in \mathbb{Z}$.
  We note that
  the only feasible values for $y$ are in $\{-2, 0, 2\}$. Hence, if $y > 2$ then both sides of the above inequality are equal to $1$ and the inequality holds. Similarly, if $y <= -2$ both
  sides of the inequality are equal to $1$ and the inequality holds. Thus, we only need to analyze the values $y=0$ and $y=2$.
  
Now we define some notation to help with rest of the proof.
Let $u = \pi'(k+1)$, $v=\pi'(k)$, and $q_i = 1-p_i$. It means $p_v < p_u$. 
  Let $\chi(i,j)$ be the event where $i$ and $j$ are indicators of decision of areas scheduled at time $k$ and $k+1$ respectively, e.g., $\chi(1,1)$ means that areas scheduled at time $k$ and $k+1$ accepted the idea, whereas $\chi(1,-1)$ implies that area scheduled at time $k$ accepted the idea, while the area scheduled at time $k+1$ rejected the idea.
  Let $B(y)$ be the event $I_{k+1} \geq I_{k-1}+y | I_{k-1} = z$ for some arbitrary $z \in \mathbb{Z}$. We consider the cases $I_{k-1} > 0$, $I_{k-1}<0$ and $I_{k-1}=0$ separately.\\
\textbf{Case 1:} $I_{k-1} = z, z>0$. We have, $B(0) = \chi(1,1) \cup \chi(1,-1) \cup \chi(-1,1)$ which is equal to the complement of  ${\chi(-1,-1)}$.
  Since we assume $z>0$, the thresholds $-c_u$ and $-c_v$ cannot be hit. Thus, $\chi(-1,-1)$ occurs only when both areas decide to reject the idea based on their respective initial acceptance probabilities. Thus, from chain rule of probability, it is the product of following four terms:
  \begin{enumerate}
  \item $Pr(z < c_u)$, i.e, the threshold rule does not apply and $u$ decides based on initial acceptance probabilities. 
  \item $u$ rejects the idea based on initial probability of rejection, $q_u$.
  \item $Pr(z - 1 < c_v)$. Given $u$ rejected the idea, $D(v)$, the decision variable for $v$ becomes $z-1$ and the threshold rule does not apply and $v$ decides based on initial acceptance probabilities.
  \item $v$ rejects the idea based on initial probability of rejection, $q_v$.
  \end{enumerate} 
  Therefore, $Pr(\chi(-1,-1)) = Pr(z < c_u)q_uPr(z - 1 < c_v)q_v$.
  Thus, $Pr(B(0);\pi) = 1-Pr(z < c_u)q_uPr(z - 1 < c_v)q_v$.
  Since, $c_u$ and $c_v$ are i.i.d random variables, we can write any probability of form $Pr(z \gtreqless c_u)$ or $Pr(z \gtreqless c_v)$ as $Pr(z \gtreqless x)$, where $x$ is an independent random variable with the same distribution as $c_u$ and $c_v$.
  Thus
  \begin{align}
    Pr(B(0);\pi) = 1-Pr(z < x)q_uPr(z - 1 < x)q_v.
  \end{align}
  Now, $Pr(\chi(1,1)) = Pr(X_u = 1|I_{k-1}=z)Pr(X_v = 1|I_k=z+1)$. Event $X_u=1$ is the union of following two non-overlapping events:
  \begin{enumerate}
  \item $z \geq c_u$; $u$ accepts the idea because of the threshold rule.
  \item $z < c_u$ and $u$ accepts the idea based on initial acceptance probability, $p_u$.
  \end{enumerate}
  Thus, $Pr(X_u = 1|I_{k-1}=z) = Pr(z \geq c_u) + Pr(z < c_u)p_u$. Similarly, $Pr(X_v = 1|I_{k}=z+1) = Pr(z+1 \geq c_v) + Pr(z+1 < c_v)p_v$. Therefore
  \begin{align}
    Pr(B(2);\pi) =&
(Pr(z \geq x) + Pr(z < x)p_u) \nonumber \\ &\times (Pr(z + 1 \geq x) + Pr(z+1 < x)p_v). \label{eq-all-yes1}
  \end{align}
  where we have replaced $c_u$ and $c_v$ by $x$ because they are i.i.d. random variables.
  We can obtain corresponding probabilities for process $\pi'$ by exchanging $p_u$ and $p_v$. Thus, $Pr(B(0);\pi) = Pr(B(0);\pi') = 1-Pr(z < x)q_uPr(z - 1 < x)q_v$. We can write $Pr(B(2);\pi')$ as follows.
  \begin{align}
  Pr(B(2);\pi') = &
(Pr(z \geq x) + Pr(z < x)p_v) \nonumber \\ &\times (Pr(z + 1 \geq x)  + Pr(z+1 < x)p_u). \label{eq-all-yes2}
\end{align}
On the other hand $Pr(z<x) \geq Pr(z+1<x)$ and $Pr(z+1 \geq x) \geq Pr(z\geq x)$. 
Comparing (\ref{eq-all-yes1}) and (\ref{eq-all-yes2}) 
along with these facts that $p_v < p_u$ and $Pr(z < x)Pr(z+1 \geq x) \geq Pr(z \geq x)Pr(z+1 < x)$, we get $ Pr(B(2);\pi) \geq Pr(B(2);\pi')$.\\
\textbf{Case 2:} $I_{k-1}=-z, z > 0$. By a similar analysis, we have
\begin{align}
  Pr(B(2);\pi) =& Pr(z < x)Pr(z-1 < x)p_up_v = Pr(B(2);\pi'), \\
  Pr(B(0);\pi) =& 1-(Pr(z \geq x) + Pr(z < x)q_u),  \nonumber\\
& \times (Pr(z + 1 \geq x) + Pr(z+1 < x)q_v),\label{eq-yn1}\\
  Pr(B(0);\pi') =& 1-(Pr(z \geq x) + Pr(z < x)q_v), \nonumber \\
& \times (Pr(z + 1 \geq x) + Pr(z+1 < x)q_u).\label{eq-yn2}
\end{align}
Comparing (\ref{eq-yn1}) and (\ref{eq-yn2}), we have $Pr(B(0);\pi) \geq Pr(B(0);\pi')$.\\
\textbf{Case 3:} $I_{k-1}=0$. We have
\begin{align}
    Pr(B(2);\pi) =& p_u(Pr(x > 1)p_v + Pr(x=1)), \label{eq-01}\\
    Pr(B(0);\pi) =& p_u + q_uPr(x > 1)p_v,  \label{eq-02} \\
    Pr(B(2);\pi') =& p_v(Pr(x > 1)p_u + Pr(x=1)), \label{eq-03}\\
    Pr(B(0);\pi') =& p_v + q_vPr(x > 1)p_u.  \label{eq-04}
  \end{align}
  By comparing (\ref{eq-01}) with (\ref{eq-03}) and (\ref{eq-02}) with (\ref{eq-04}), we see that $Pr(B(2);\pi) \geq Pr(B(2);\pi')$ and $Pr(B(0);\pi) \geq Pr(B(0);\pi')$ respectively. Thus, $Pr(I_{k+1} \geq I_{k-1} + x| I_{k-1}; \pi) \geq Pr(I_{k+1} \geq I_{k-1} + x| I_{k-1}; \pi'), \forall x \in \mathbb{Z}$.
\end{proof}

%% file: appendix.tex
\input{example.tex}


\input{switch.tex}

\input{hardness}

\input{given-schedule.tex}

\input{adaptive}

\input{missing}

%% file: example.tex
\section{Examples}
\label{sec:example}
\begin{example}
\label{example:1}
Consider a society with $3$ areas and $3$ types. The planner prior is as follows. Initial acceptance probabilities of areas $1$, $2$, and $3$ are $0.2$, $0.5$, and $0.8$ respectively. Thresholds of areas $1$, $2$, and $3$ are $1$, $2$, and $3$ respectively  (See Figure \ref{fig:example1}). Consider spreading strategy $\pi=(1, 2, 3)$. People in area $1$ accept the idea with probability $p_1=0.2$. Threshold of area $2$ is $2$. It means people in area $2$ decide based on initial rule and accept the idea with probability $p_2=0.5$.  Threshold of area $3$ is $3$. Thus, people in area $3$ decide based on initial rule as well and accept the idea with probability $p_3=0.8$. Therefore, the expected number of adopters for spreading strategy $\pi$ is $p_1+p_2+p_3=1.5$. In order to see the impact of an optimal spreading strategy consider spreading strategy $\pi'=(3,1,2)$. 
People in area $3$ accept the idea with probability $p_3=0.8$. Threshold of area $1$ is $1$. It means the decision of people in area $1$ is correlated to the decision of people in area $3$. In other word, people in area $1$ follow the decision of people in area $3$. Thus, there are two possible scenarios. First, both areas $3$ and $1$ accept the idea. The probability of this scenario is $p_3=0.8$. The second scenario is that both areas $3$ and $1$ reject the idea. The probability of the second scenario is $1-p_3=0.2$. In both scenario the threshold of area $2$ is hit. Hence, area $2$ will accept the idea with probability $p_3=0.8$. Therefore, the expected number of adopters for spreading schedule $\pi'$ is $3p_3=2.4$.  
%
\begin{figure}
\centering
\input{Fig-Ex1}
\caption{A society with $3$ areas. The expected number of adopters for spreading strategy $\pi=(1, 2, 3)$ is $1.5$. The expected number of adopters for spreading strategy $\pi'=(3, 1, 2)$ is $2.4$.}
\label{fig:example1}
\end{figure}
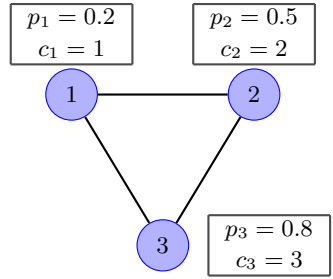
\end{example}

\begin{example}
\label{example:2}
At the first glance, it seems a greedy approach leads us to find the best non-adaptive spreading strategy. The greedy approach is to first schedule a node with the highest probability of adopting. We find a counter-example for this greedy approach with a society with $3$ areas. 

Consider a society with $3$ areas and $3$ types. Area $1$ has threshold $1$ and areas $2$ and $3$ have threshold $2$. Initial acceptance probabilities are $p_1 > p_2 > p_3=0$ (See Figure \ref{fig:example2}). The greedy approach leads us to spreading strategy $\pi=(1,2,3)$. Assume the planner uses spreading strategy $\pi$. The probability that people in area $1$ accept the idea is $p_1$. The threshold for area $2$ is $2$. Hence, they decide based on initial rule. It means the probability that people in area $2$ accept the idea is $p_2$. At last, if both area $1$ and $2$ accept the idea then people in area $3$ accept the idea with probability $p_1p_2$ based on threshold rules . Otherwise, they reject it because $p_3=0$, i.e., area $3$ has an initial preference of $\mathcal{N}$ for sure. Thus, the expected number of adopter is $p_1 + p_2 + p_1p_2$. 
Now, assume the planner uses spreading strategy $\pi'=(2,1,3)$.  Area $2$ accepts the idea with probability $p_2$. The threshold of area $1$ is $1$. It means area $1$ is a follower of area $2$ under spreading strategy $\pi'$. Hence, there are two possibilities. Both areas $1$ and $2$ accept the idea with probability $p_2$ or both areas $1$ and $2$ reject the idea with probability $1-p_2$. In both cases area $3$ decides based on the threshold rule. Therefore, there are $3$ adopters with probability $p_2$ or all areas reject the idea with probability $1-p_2$. Hence, the expected number of adopter is $3p_2$ for spreading strategy $\pi'$. One can check spreading strategy $\pi'$ is better that $\pi$ for various probabilities $p_1$ and $p_2$, e.g., $p_1=0.4$ and $p_2=0.3$ or $p_1=0.8$ and $p_2=0.7$.
\begin{figure}
\centering
\input{Fig-Ex2}
\caption{A society with $3$ areas. The expected number of adopters for spreading strategy $\pi=(1, 2, 3)$ is $p_1+p_2+p_1p_2$. The expected number of adopters for spreading strategy $\pi'=(2, 1, 3)$ is $3p_2$.}
\label{fig:example2}
\end{figure}
\end{example}

\begin{example}
\label{example:3}
The result of Theorem \ref{sec:bounds-spread-weak} leads us to the following conjecture for the partial propagation setting. 
\begin{quote}
``Consider an arbitrary non-adaptive spreading strategy in the partial propagation setting. If all initial acceptance probabilities are greater/less than $\frac{1}{2}$, then adding an edge to the graph helps/hurts promoting the new product.".
\end{quote}
This conjecture has several consequences, e.g., a complete graph is the best graph for spreading a new idea when initial acceptance probabilities are greater than $\frac{1}{2}$. This eventuates directly Theorem \ref{sec:bounds-spread-weak}.
Surprisingly, this conjecture does not hold. We present an example with the same initial acceptance probabilities of less than $\frac{1}{2}$ such that adding a relationship between two areas increases the expected number of adopters. 

Consider a society with $4$ areas and only one type. Initial acceptance probabilities and thresholds for all areas are $p$ and $1$ respectively. Consider spreading strategy $\pi=(1, 2, 3, 4)$ and a society which is represented by graph $G$ (See Figure \ref{fig:example3}). Areas $1$, $2$, and $3$ decide about the idea independently and accept it with probability $p$. Threshold of area $4$ is $1$. Hence, people in area $4$ accept the idea if there are at least two adopters so far. Therefore, area $4$ accept the idea with probability $3p^2(1-p)+p^3$ and the expected number of adopters is $3p+3p^2(1-p)+p^3$.
Assume influences also propagate between area $1$ and $2$. In this case the society is represented by graph $G'$ (See Figure \ref{fig:example3}).
Threshold of area $2$ is $1$. Hence, area $2$ is a follower of area $1$ under spreading strategy $\pi$. Thus, there are two possibilities when area $2$ is scheduled. Both area $1$ and $2$ accept the idea with probability $p$ or both reject it with probability $1-p$. 
Area $3$ decide independently and accept the idea with probability $p$. Threshold of area $4$ is $1$. Thus, area $4$ is also a follower of both area $1$ and $2$.  Therefore, the expected number of adopter is $4p$ in this case.
One can check $3p+3p^2(1-p)+p^3$ is greater than $4p$ if and only if $0.5<p<1$. It means when $p<0.5$ (resp., $p>0.5$) the number of adopters increases (resp., decreases) by adding a relation to the society. 
\begin{figure}
\centering
\input{Fig-Ex3}
\caption{This figure represents a partial propagation setting with $4$ areas. All Thresholds are equal to $1$ and all initial acceptance probabilities are $p$. The expected number of adopters for spreading strategy $\pi=(1, 2, 3,4)$ is $3p+3p^2(1-p)+p^3$ for a society which is represented by graph $G$. The expected number of adopters for spreading strategy $\pi=(1, 2, 3,4)$ is $4p$ for a society which is represented by graph $G'$. Note that $3p^2(1-p)+p^3$ is greater than $p$ if and only if $0.5<p<1$}
\label{fig:example3}
\end{figure}
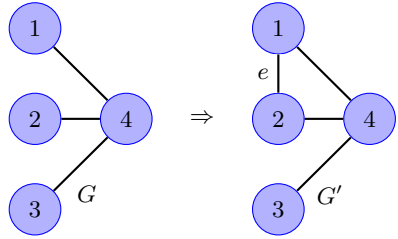
\end{example}

%% file: Fig-Ex1.tex
\tikzstyle{H-node}=[rectangle,draw=black,fill=white!30,inner sep=1.3mm]
\tikzstyle{B-node}=[circle,draw=blue,fill=blue!30,inner sep=2.4mm]
\tikzstyle{G-node}=[circle,draw=green,fill=green!30,inner sep=1.3mm]
\tikzstyle{R-node}=[circle,draw=red,fill=red!30,inner sep=2.5mm]
\tikzstyle{W-node}=[rectangle,draw=white,fill=white!30,inner sep=2.5mm]
\tikzstyle{test-node}=[circle,draw=black,fill=black,inner sep=.2mm]

\tikzstyle{bl0} = [draw=black, thick, dashed]   
\tikzstyle{b9} = [draw=black, thick]   
\tikzstyle{bl1} = [->, draw=black]   
\tikzstyle{bl2} = [draw=black!70,thick]   
\tikzstyle{bl3} = [draw=black,thick, dotted]   

\tikzstyle{br0} = [draw=brown, dashed]   
\tikzstyle{br1} = [->, draw=brown]   
\tikzstyle{br2} = [->, draw=brown,thick]   

\tikzstyle{red0} = [draw=red, thick, dashed]   
\tikzstyle{red1} = [draw=red]   
\tikzstyle{red2} = [draw=red,thick]   

\tikzstyle{gr0} = [draw=green, thick, dashed]   
\tikzstyle{gr1} = [draw=green]   
\tikzstyle{gr2} = [draw=green,thick]   
\tikzstyle{gr4} = [draw=green,semithick,rounded corners]   

\begin{tikzpicture}[scale=0.4][domain=0:8]
\draw (8,7.5) node[W-node,label=center:{$p_2 = 0.5$}] (p_2) {};
\draw (8,6.5) node[W-node,label=center:{$c_2 = 2$}] (c_2) {};
\draw[bl2] (6, 8) -- (6, 6) ;
\draw[bl2] (6, 6) -- (10, 6) ;
\draw[bl2] (10, 6) -- (10, 8) ;
\draw[bl2] (10, 8) -- (6, 8) ;
\draw (8,5) node[B-node,label=center:$2$] (u_2) {};

\draw (2,7.5) node[W-node,label=center:{$p_1 = 0.2$}] (p_1) {};
\draw (2,6.5) node[W-node,label=center:{$c_1 = 1$}] (c_1) {};
\draw[bl2] (0, 8) -- (0, 6) ;
\draw[bl2] (0, 6) -- (4, 6) ;
\draw[bl2] (4, 6) -- (4, 8) ;
\draw[bl2] (4, 8) -- (0, 8) ;
\draw (2,5) node[B-node,label=center:$1$] (u_1) {};

\draw (8.5,0.5) node[W-node,label=center:{$p_3 = 0.8$}] (p_3) {};
\draw (8.5,-.5) node[W-node,label=center:{$c_3 = 3$}] (c_3) {};
\draw[bl2] (6.5, 1) -- (6.5, -1) ;
\draw[bl2] (6.5, -1) -- (10.5, -1) ;
\draw[bl2] (10.5, -1) -- (10.5, 1) ;
\draw[bl2] (10.5, 1) -- (6.5, 1) ;
\draw (5,0) node[B-node,label=center:$3$] (u_3) {};

\draw[b9] (u_1) -- (u_2) ;
\draw[b9] (u_1) -- (u_3) ;
\draw[b9] (u_2) -- (u_3) ;
\end{tikzpicture}

%% file: Fig-Ex2.tex
\tikzstyle{H-node}=[rectangle,draw=black,fill=white!30,inner sep=1.3mm]
\tikzstyle{B-node}=[circle,draw=blue,fill=blue!30,inner sep=2.4mm]
\tikzstyle{G-node}=[circle,draw=green,fill=green!30,inner sep=1.3mm]
\tikzstyle{R-node}=[circle,draw=red,fill=red!30,inner sep=2.5mm]
\tikzstyle{W-node}=[rectangle,draw=white,fill=white!30,inner sep=2.5mm]
\tikzstyle{test-node}=[circle,draw=black,fill=black,inner sep=.2mm]

\tikzstyle{bl0} = [draw=black, thick, dashed]   
\tikzstyle{b9} = [draw=black, thick]   
\tikzstyle{bl1} = [->, draw=black]   
\tikzstyle{bl2} = [draw=black!70,thick]   
\tikzstyle{bl3} = [draw=black,thick, dotted]   

\tikzstyle{br0} = [draw=brown, dashed]   
\tikzstyle{br1} = [->, draw=brown]   
\tikzstyle{br2} = [->, draw=brown,thick]   

\tikzstyle{red0} = [draw=red, thick, dashed]   
\tikzstyle{red1} = [draw=red]   
\tikzstyle{red2} = [draw=red,thick]   

\tikzstyle{gr0} = [draw=green, thick, dashed]   
\tikzstyle{gr1} = [draw=green]   
\tikzstyle{gr2} = [draw=green,thick]   
\tikzstyle{gr4} = [draw=green,semithick,rounded corners]   

\begin{tikzpicture}[scale=0.4][domain=0:8]
\draw (8,7.5) node[W-node,label=center:$p_2$] (p_2) {};
\draw (8,6.5) node[W-node,label=center:{$c_2 = 2$}] (c_2) {};
\draw[bl2] (6, 8) -- (6, 6) ;
\draw[bl2] (6, 6) -- (10, 6) ;
\draw[bl2] (10, 6) -- (10, 8) ;
\draw[bl2] (10, 8) -- (6, 8) ;
\draw (8,5) node[B-node,label=center:$2$] (u_2) {};

\draw (2,7.5) node[W-node,label=center:$p_1$] (p_1) {};
\draw (2,6.5) node[W-node,label=center:{$c_1 = 1$}] (c_1) {};
\draw[bl2] (0, 8) -- (0, 6) ;
\draw[bl2] (0, 6) -- (4, 6) ;
\draw[bl2] (4, 6) -- (4, 8) ;
\draw[bl2] (4, 8) -- (0, 8) ;
\draw (2,5) node[B-node,label=center:$1$] (u_1) {};

\draw (8.5,0.5) node[W-node,label=center:{$p_3 = 0$}] (p_3) {};
\draw (8.5,-.5) node[W-node,label=center:{$c_3 = 2$}] (c_3) {};
\draw[bl2] (6.5, 1) -- (6.5, -1) ;
\draw[bl2] (6.5, -1) -- (10.5, -1) ;
\draw[bl2] (10.5, -1) -- (10.5, 1) ;
\draw[bl2] (10.5, 1) -- (6.5, 1) ;
\draw (5,0) node[B-node,label=center:$3$] (u_3) {};

\draw[b9] (u_1) -- (u_2) ;
\draw[b9] (u_1) -- (u_3) ;
\draw[b9] (u_2) -- (u_3) ;
\end{tikzpicture}

%% file: Fig-Ex3.tex
\tikzstyle{H-node}=[rectangle,draw=black,fill=white!30,inner sep=1.3mm]
\tikzstyle{B-node}=[circle,draw=blue,fill=blue!30,inner sep=2.4mm]
\tikzstyle{G-node}=[circle,draw=green,fill=green!30,inner sep=1.3mm]
\tikzstyle{R-node}=[circle,draw=red,fill=red!30,inner sep=2.5mm]
\tikzstyle{W-node}=[rectangle,draw=white,fill=white!30,inner sep=2.5mm]
\tikzstyle{test-node}=[circle,draw=black,fill=black,inner sep=.2mm]

\tikzstyle{bl0} = [draw=black, thick, dashed]   
\tikzstyle{b9} = [draw=black, thick]   
\tikzstyle{bl1} = [->, draw=black]   
\tikzstyle{bl2} = [draw=black!70,thick]   
\tikzstyle{bl3} = [draw=black,thick, dotted]   

\tikzstyle{br0} = [draw=brown, dashed]   
\tikzstyle{br1} = [->, draw=brown]   
\tikzstyle{br2} = [->, draw=brown,thick]   

\tikzstyle{red0} = [draw=red, thick, dashed]   
\tikzstyle{red1} = [draw=red]   
\tikzstyle{red2} = [draw=red,thick]   

\tikzstyle{gr0} = [draw=green, thick, dashed]   
\tikzstyle{gr1} = [draw=green]   
\tikzstyle{gr2} = [draw=green,thick]   
\tikzstyle{gr4} = [draw=green,semithick,rounded corners]   

\begin{tikzpicture}[scale=0.4][domain=0:8]
\draw (5,5) node[B-node,label=center:$4$] (u_4) {};
\draw (2,5) node[B-node,label=center:$2$] (u_2) {};
\draw (2,2) node[B-node,label=center:$3$] (u_3) {};
\draw (2,8) node[B-node,label=center:$1$] (u_1) {};

\draw (3.7,2.5) node[W-node,label=center:$G$] (G) {};
\draw[b9] (u_1) -- (u_4) ;
\draw[b9] (u_3) -- (u_4) ;
\draw[b9] (u_2) -- (u_4) ;

\draw (7.5,5) node[W-node,label=center:$\Rightarrow$] (RA) {};

\draw (13,5) node[B-node,label=center:$4$] (v_4) {};
\draw (10,5) node[B-node,label=center:$2$] (v_2) {};
\draw (10,2) node[B-node,label=center:$3$] (v_3) {};
\draw (10,8) node[B-node,label=center:$1$] (v_1) {};

\draw (11.7,2.5) node[W-node,label=center:$G'$] (GG) {};
\draw (9.5,6.5) node[W-node,label=center:$e$] (GG) {};
\draw[b9] (v_1) -- (v_4) ;
\draw[b9] (v_3) -- (v_4) ;
\draw[b9] (v_2) -- (v_4) ;
\draw[b9] (v_1) -- (v_2) ;

\end{tikzpicture}

%% file: switch.tex
\section{Type Switching Approach}
\label{sec:switch}
Consider a society with a constant number of types. One approach that might work is an algorithm that finds an optimal spreading strategy allowing for only a constant number of switches between types in a spreading strategy. 
We note that areas of the same type are identical from point of view of scheduling a cascade. Thus, any non-adaptive spreading strategy can be specified by specifying types of areas rather than the areas themselves. Let $\tau$ be the mapping between an area and its type. That is $\tau(i)$ is the type of area $i$. Let $\lambda$ be sequence of types for a given spreading strategy. Specifically, $\lambda$ is a vector whose $k^{th}$ component, $\lambda(k) = \tau(\pi(k))$. A switch is any position $k$ in the sequence $\lambda$ such that $\lambda(k) \ne \lambda(k+1)$. As an example, consider a society with four areas with two areas of type $1$ and two areas of type $2$. Then the type sequence $\lambda = (1, 1, 2, 2)$ has a switch at position $2$ whereas $\lambda_2 = (1, 2, 1,2)$ has switches at positions $1$, $2$ and $3$. We define a $\sigma$-switch spreading strategy as a non-adaptive spreading strategy that has at most $\sigma$ switches, where $\sigma$ is a constant independent of input size. We now prove that no algorithm whose output is a $\sigma$-switch spreading strategy can be optimal.

\switch*
\begin{proof}
  The proof outline is as follows. We construct an instance of problem with $2n$ areas with two types, the number of areas of both types being $n$, for which an optimal spreading strategy alternates between these types. Lets call this instance $S$ and lets call this strategy $\pi$. We prove that the expected number of adopters achieved by this optimal strategy is upper bound on number of acceptors for any input instance with areas of these two types, whatever be the number of areas of both types, given that total number of areas is $2n$, e.g., the number of areas of one type can be $n_1$ and the other type $2n-n_1$ for any integer $n_1$ between $0$ and $2n$ and no strategy for this instance can exceed the expected number of adopters achieved by $\pi$ for the instance of problem with $n$ areas of each type. We then show that any $\sigma$-switch strategy for instance $S$ of problem can be improved by changing type of one of the areas. Since, the optimal value achieved by this new strategy cannot be greater than strategy $\pi$ on instance $S$, no $\sigma$-switch strategy can be optimal.

Consider an instance with two types $\gamma_1 = (P,1)$ and $\gamma_2=(P,2)$ where $P > \frac{1}{2}$, the total number of areas is $2n$ and the number of areas of types $\gamma_1$ and $\gamma_2$ is $n$ each. Let $\pi$ be a spreading strategy for which the type sequence of areas is given by $\lambda = (\gamma_1, \gamma_2,\ldots,\gamma_1,\gamma_2)$, i.e., every area at odd position is of type $\gamma_1$ and every area at even position is of type $\gamma_2$. Let the expected number of areas which accept the idea for this spreading strategy be $\alpha$. Now consider an instance where the total number of areas is the same but the number of areas of type $\gamma_1$ is $n_1$ and number of areas of type $\gamma_2$ is $2n-n_1$ for some arbitrary natural number $n_1$ such that $0 \leq n_1 \leq n_2$. For this instance, let the expeted number of areas which accept the idea given an optimal spreading strategy be $\beta$. We now prove that $\alpha \geq \beta$.  If we have no restriction on the number of areas of each type, then for any $t = 0\mod 2$, the areas to be scheduled at time $t+1$ and $t+2$ can be of types $(\gamma_1, \gamma_1)$, $(\gamma_1, \gamma_2)$, $(\gamma_2, \gamma_1)$ or $(\gamma_2, \gamma_2)$. We prove that $\alpha \geq \beta$ by proving that it is better to schedule areas of type $\gamma_1$ and $\gamma_2$ at times $t+1$ and $t+2$ respectively. If $|I_t| \geq 2$, then we are indifferent between all spreading strategies because in this case all the areas will decide based on the threshold rule. Thus, if we can prove that $(\gamma_1, \gamma_2)$ is a best choice for types at times $t+1$ and $t+2$ when $|I_t| < 2$, we are done. Since $t$ is even, the only feasible value of $|I_t| \leq 2$ is $I_t = 0$. Thus, this is the only case we need to analyze. Let $\rho$ be the tuple of types of areas scheduled at times $t+1$ and $t+2$. Let $\chi$ be the tuple indicating decisions of areas scheduled at times $t+1$ and $t+2$. Now we analyze the probabilties with which the four possible values of $\chi$ are realized for each of the four possible values of $\rho$ when $I_t = 0$. Let number of areas to be scheduled after time $t$ be $m$.\\
  \textbf{Case 1:} $\rho = (\gamma_1, \gamma_1)$ or $(\gamma_2, \gamma_1)$\\
  In this case, the first area decides based on its initial acceptance probability and the second area follows the decision of the first area.
  \begin{align*}
    Pr(\chi = (1,1)) &= P\\
    Pr(\chi = (1,-1)) &= 0\\
    Pr(\chi = (-1,1)) &= 0\\
    Pr(\chi = (-1,-1)) &= 1-P
  \end{align*}
  The expected number of areas which accept the idea after time $t$ in this case is $mP$, as all areas follow the decision of area scheduled at time $t+1$.\\
  \textbf{Case 2:} $\rho = (\gamma_1, \gamma_2)$ or $(\gamma_2, \gamma_2)$\\
  In this case, both the areas decide based on their initial acceptance probability.
  \begin{align}
    Pr(\chi = (1,1)) &= P^2\label{eq:9}\\
    Pr(\chi = (1,-1)) &= P(1-P)\label{eq:10}\\
    Pr(\chi = (-1,1)) &= P(1-P)\label{eq:11}\\
    Pr(\chi = (-1,-1)) &= (1-P)^2
  \end{align}
  From (\ref{eq:9}), with probability $P^2$, all areas after time $t$ will accept the idea. If for any time $t'$, we are given that $I_{t'} = 0$, then we can treat the subsequent areas as the starting point of a new spreading strategy. Thus, if $I_{t+2}=0$, then from Theorem \ref{sec:bounds-spread-weak} (given that $P > \frac{1}{2}$), the expected number of adopters for any future spreading strategy is at least $(m-2)P$. Hence, from (\ref{eq:10}) and (\ref{eq:11}), with probability $2P(1-P)$ the expected number of areas that will accept after time $t$ is at least $1 + (m-2)P$. Therefore, in this case, the expected number of areas that accept after time $t$ is at least $mP^2 + 2P(1-P)(1 + (m-2)P)$. Thus, we are done if we prove that $mP^2 + 2P(1-P)(1 + (m-2)P)$ is greater than $mP$.
  \begin{align*}
    mP^2 + 2P(1-P)(1 + (m-2)P) - mP = P(1-P)(-m + 2(1+(m-2)P))
  \end{align*}
  Thus, it is enough to prove that $2(1+(m-2)P) - m > 0$. We have:
  \begin{align*}
    2(1+(m-2)P) - m = (2P-1)(m-2)
  \end{align*}
  Since $P > \frac{1}{2}$, $2P-1 > 0$. Thus, for all $m > 2$, it is strictly better to schedule an area of type $\gamma_2$ at time $t+2$. If an area of type $\gamma_2$ is scheduled at time $t+2$, then it is equivalent to schedule an area of either type at time $t+1$. Thus, given that there is at least one more area to follow at time $t+3$, it is best to schedule areas of type $\gamma_1$ and $\gamma_2$ respectively at times $t+1$ and $t+2$ at any arbitrary time $t = 0 \mod 2$. Also, such a schedule is strictly better, all other things begin same, than the schedule where, areas of type $\gamma_1$ are scheduled at times $t+1$ and $t+2$. This fact is important as we use this later in the proof. If there are no more areas to follow, then we are indifferent to all the four options. Hence, the expected number of adopters achieved by $\pi$ is an upper bound on number of acceptors for any input instance with areas of these two types whatever be the number of areas of both types

The final part of this proof is by contradiction. Let the the number of areas in the input instance of problem be $2n$ with $n$ areas each of types $\gamma_1 = (P,1)$ and $\gamma_2 = (P,2)$. Consider a $\sigma$-switch strategy. Choose $n \geq 4(\sigma+1)$. Thus, every $\sigma$-switch strategy will have at least four consecutive areas of type $\gamma_1$. Let a $\sigma$-switch strategy, $\pi'$, be an optimal one. Therefore, there will exist a time $t$ in $\pi'$ such that $t = 0\mod 2$, $\tau(\pi'(t+1)) = \gamma_1$, $\tau(\pi'(t+2)) = \gamma_1$ and at least one more area will be scheduled after time $t+2$. As explained earlier, the expected number of adopters in this case is strictly less than expected number of adopters if we schedule an area of type $\gamma_2$ at time $t+2$, which, as proved above, is at most the expected number of adopters for a strategy with type sequence $(\gamma_1, \gamma_2, \ldots, \gamma_1, \gamma_2)$. Therefore, strategy $\pi$ is not optimal. This is a contradiction and no $\sigma$-switch strategy can be optimal for the given instance.
\end{proof}

%% file: hardness.tex
\section{Hardness Result}
\label{sec:hardness}

We prove that problem of computing expected number of adopters for a given spreading strategy in the partial propagation setting is $\#P$-complete. This result applies even when the input graphs are planer with a maximum degree of $3$ and have only $4$ different types of vertices. We prove this by reduction from a version of the network reliability problem that is known to be $\#P$-complete (\cite{P86}). In the network reliability problem, a directed graph $G$ and probability $0 \leq p \leq 1$ are given. Nodes fail independently with probability $1-p$. Therefore, each node is present in the surviving subgraph with probability $p$. We achieve the reduction by simulating the $s-t$ network reliability problem by designing an instance of cascade scheduling problem where, probability of an area $v$ accepting an idea is exactly equal to a path existing in the surviving sub-graph from the source to vertex $v$. Before proceeding to details of the proof, we give some definitions below.

\begin{defn}
  Given a directed graph $G$ with source $s$, terminal $t$, and a probability $1-p, 0 \leq p < 1$ of nodes failing independently,
  the \emph{$(s,t)$-connectedness reliability} of $G$, $R(G, s, t; p)$, is defined as the probability that there is at least one path 
  from $s$ to $t$ such that none of the vertices falling on the path have failed.
\end{defn}
\begin{defn}
  \emph{AST} is the problem of computing $R(G, s, t; p)$ when $G$ is an acyclic directed
  $(s, t)$-planar graph with each vertex having degree at most three. We denote an instance
  of AST on graph $G$ as $AST(G,s,t,p)$.
\end{defn}
\begin{defn}
  Given an influence spread process, $S=(G,\mathbf{c},\mathbf{p},\pi)$ on $G$ with a source node $s$ and a target node $t$, \emph{IST}
  is the problem of computing $Pr(X_t = 1; S)$ given that $\pi(1) = s$ and $Pr(X_s=1)=1$. We denote an instance of IST by
  $IST(G, \mathbf{c}, \mathbf{p}, \pi, s, t)$.
\end{defn}
We will reduce an instance of AST to an instance of IST (Probability of Influence Spread to T).

Given an instance of AST, $AST(G=(V,E),s,t,p)$ we now construct an instance of IST, $IST(G'=(V' ,E'), \mathbf{c}, \mathbf{p}, \pi, s, t)$ for which $R(G,s,t;p) = Pr(X_t=1)$. Let $d^{in}_v$ be the indegree of $v \in V$ in $G$. For every vertex $v \in V - \{s\}$, we add three vertices to graph $G'$. Lets denote them by $b_v$, the blocking vertex of $v$, $f_v$, the forwarding vertex for $v$ and $v'$, which corresponds to the original vertex $v$. The rationale for nomenclature will become apparent later. For every edge $(u,v)$ in $E$, we add an edge $\{u', b_v\}$ in $E'$. In addition, we add edges $\{b_v, v'\}$ and $\{f_v, v'\}$ to $E'$. The acceptance probabilities and thresholds are set as follows: $p_{v'} = 0, p_{f_v} = p, p_{b_v}=1 \fa v \in V - \{s\}, p_{s'} = p$. $c_v = 2, c_{b_v} = d^{in}_v  \fa v \in V - \{s\}$. Threshold $c_{s'}$ is irrelevant and can be any arbitrary value greater than $0$ since it is the first vertex to be scheduled. Thresholds $c_{f_v}$ can also be any arbitrary value greater than $0$ since no neighbor of $f_v$ is scheduled before $f_v$. Let $\pi':V\mapsto V$ be any topological ordering on $V$ where, $s$ is the first node and $t$ is the last node. Then $\pi$ is constructed as follows:
\begin{align*}
  \pi^{-1}(s') =& 1\\
  \pi^{-1}(v') =& 3\pi'^{-1}(v) - 2 \fa v \in V - \{s\}\\
  \pi^{-1}(b_v) =& 3\pi'^{-1}(v) - 4 \fa v \in V - \{s\}\\
  \pi^{-1}(f_v) =& 3\pi'^{-1}(t) - 3 \fa v \in V - \{s\}
\end{align*}
The above construction of $\pi$ can be interpreted as follows. Source remains the first vertex to be scheduled. A vertex $v$ is split into three vertices --- $v'$, $b_v$ and $f_v$. In place of $v$, these three vertices are consecutively scheduled in order $b_v$, $f_v$ and $v'$, e.g., if $\pi' = (s,v,t)$, then $\pi=(s', b_v, f_v, v', b_t, f_t, t')$. 

Let $IS$ be the influence spread process $(G', \mathbf{c}, \mathbf{p}, \pi)$. Now, we prove the following lemmas which relate the probability of existence of a path of operative vertices between $s$ and $v$ in $G$ and the probability that area $v$ accepts the idea in the influence spread process $IS$.

\begin{figure}
\centering
\input{Figure1}
\caption{Reduction from Network Reliability on a DAG to Computing Expected Number of Influenced Nodes -- The diagram on left is a part of DAG with probability of failure of each node equal to $(1-p)$. The diagram on right is corresponding part of graph that represents an influence spread stochastic process the models the given network reliability problem where $p_{b_v} =1$, $c_{b_v}=d$, $p_{f_v}=p$,$p_{v'}=0$, and $c_{v'}=2$.}
\label{fig:generated-runtime}
\end{figure}
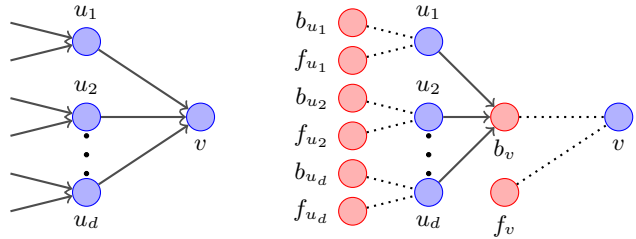
We first prove that computing the expecte number of vertices in graph to which $s$ has a path with operating vertices is $\#P$-complete. We then use this to prove the main theorem.
\begin{lem}
  \label{sec:hardness-result}
  Consider an instance of AST, $AST(G=(V,E), s,t,p)$. Then computing the expected number of vertices in graph to which $s$ has a path with operating vertices is $\#P$-complete.
\end{lem}
\begin{proof}
  Let $a(G,s)$ be the expected number of vertices in the graph to which $s$ has a path with operating vertices in $G$.
  Let $b(G, s,t)$ be probability that there is a path of operating vertices from $s$ to $t$ in $G$. We note that $t$ has no outgoing edges.
  Lets assume that $a(G,s)$ can be computed in time polynomial in $|G|$.
  Let $G' = G - \{t\}$. Deletion of $t$ does not change probability of survival of any path whose destination is not $t$. Therefore $a(G',s) = \sum_{u \in V-\{t\}}b(G,s,u)$. Thus, $a(G,s) - a(G',s) = b(G,s,t)$. This is a contradiction because this implies that $b(G,s,t)$ can be computed in time polynomial in $|G|$.
\end{proof}

The proof of the main theorem of this section is organized as follows. We first prove that 
the probability of an area $v'$ accepting an idea is exactly equal to probability of a path existing from $s$ to $v$. Then, we use this fact along with Lemma \ref{sec:hardness-result} to prove the main result.
\hardness*
\begin{proof}
Let $AST(G=(V,E),s,t,p)$ be an instance of AST problem. Let $S(G'=(V',E'), \mathbf{c}, \mathbf{p}, \pi)$ be an influence spread process with $G', c_v, p_v$ and $\pi$ as defined above. Then an area $v \ne s,t$ accepts the idea with probability $p$ iff at least one of its predecessors in $G$ also accepts the idea.

Let $P(v)$ be the set of predecessors of $v$ in $G$. We note that in $IS$, by construction of $\pi$ and $G'$, vertices in $P(v)$ are exactly the neighbors of $b_v$ that are scheduled before $b_v$. Area $b_v$ is immediately followed by $f_v$ and $f_v$ by $v$. Also, by construction of $G'$, $b_v$ and $f_v$ are neighbors of $v$ and $v$ has no other neighbors. Area $f_v$'s only neighbor is $v$.

If no vertex in $P(v)$ accepts the idea, then $D(b_v) = -d^{in}_v = -c_{b_v}$ and thus, $Pr(b_v=-1| \text{ no vertex in $P(v)$ accepts the idea }) = 1$ and therefore, $b_v$ rejects the idea. Since, threshold of $v$ is $c_v=2$, $v$ decides based on threshold if and only if both its neighbors either accept or reject the idea. Therefore if $b_v$ rejects the idea, then if $f_v$ accepts the idea, then $v$ does not accept the idea because it decides to reject the idea based on its initial acceptance probability as $p_v = 0$.  If $X_{f_v}=-1$, then also $v$ does not accept the idea because it reject the idea based on threshold rule, because both its neighbors rejected this idea. Thus, if none of the vertices in $P(v)$ accept the idea then $v$ does not accept the idea.

If any area in $P(v)$ accepts the idea then $-c_{b_v} = -d^{in}_v < D(b_v) < d^{in}_v = c_{b_v}$ and $b_v$ accepts the idea because its initial acceptance probability, $p_{b_v} = 1$. Now, if $f_v$ accepts the idea then $v$ also accepts because $c_v = 2$ and if $f_v$ rejects the idea, then $v$ does not accept the idea because it decides to reject it on basis of its initial acceptance probability, $p_v = 0$. Since, no neighbor of $f_v$ is scheduled before $f_v$, $f_v$ accepts the idea independently at random with its initial acceptance probability $p_{f_v} = p$. Therefore, given that at least one vertex in set $P(v)$ accepts the idea, $v$ accepts the idea with probability $p$.  

Now, by principal of deferred decisions, process of finding a path of operating vertices from $s$ to $t$ in the network reliability problem, can be simulated as follows. Let $\pi$ be any topological ordering on vertices of $G$. Let $L(i)$ be the $i^{th}$ layer (excluding layer containing just the source vertex, $s$) in topologically sorted $G$. Then probability that a path to $u \in L(1)$ exists is $p$ because we let each vertex in this layer fail independently with probability $1-p$. For vertex $v$ in any subsequent layer, if there exists a path to any of vertices in $P(v)$, the set of predecessors of $v$, then we let $v$ fail independently with probability $1-p$. If no path to any of predecessors of $v$ exists, then no path to $v$ can exist and it is immaterial whether $v$ fails or not. Thus, we let $v$ fail with probability $1$. As explained above, this is exactly the process simulated by $IS(G', c_v, p_v, \pi)$. Thus, computing $Pr(X_t=1)$ is $\#P$-complete. 

However, we need to prove hardness of computing $\Lambda = \sum_{u \in V'}Pr(X_{u}=1)$. If we can prove that from $\Lambda$ we can compute the expected number of vertices in graph to which $s$ has a path, say $\alpha= \sum_{v \in V}Pr(X_{v'}=1)$, then from Lemma \ref{sec:hardness-result}, we are done.

Since $\forall v \in V, Pr(X_{v'} = 1) = Pr(X_{b_v}=1)\cdot Pr(X_{f_v} = 1) = Pr(X_{b_v}=1)\cdot p$
and $Pr(X_{f_v}) = p$, we have:
\begin{align*}
  \Lambda =& \sum_{v \in V}(Pr(X_{v'}=1)+Pr(X_{b_v}=1)+Pr(X_{f_v}=1))
  =& \sum_{v \in V}(Pr(X_{v'}=1)+\frac{Pr(X_{v'}=1)}{p}+p)
\end{align*}
From above, we can easily compute $\alpha$. Hence, the claim follows.
\end{proof}

We note that AST is $\#P$-complete even when degrees of vertices of the input graph is constrained to be $3$. Thus, indegree of a node (through which a path from $s$ to $t$ can pass) has to be $1$ or $2$. If $p$ is the survival probability of a vertex in the AST problem instance, then the possible types of areas in the corresponding instance of IST are in $\{(1,1),(1,2),(p,1),(0,2)\}$, where the first two types correspond to blocking nodes in $G$, the forwarding nodes are of type $(p,1)$ and the vertices corresponding to original vertices are of type $(0,2)$. Thus, IST is hard on graphs with maximum degree constrained to $3$ and number of types constrained to $4$.

%% file: Figure1.tex
\tikzstyle{H-node}=[rectangle,draw=red,fill=red!30,inner sep=1.3mm]
\tikzstyle{B-node}=[circle,draw=blue,fill=blue!30,inner sep=1.3mm]
\tikzstyle{G-node}=[circle,draw=green,fill=green!30,inner sep=1.3mm]
\tikzstyle{R-node}=[circle,draw=red,fill=red!30,inner sep=1.3mm]
\tikzstyle{test-node}=[circle,draw=black,fill=black,inner sep=.2mm]

\tikzstyle{bl0} = [draw=black, thick, dashed]   
\tikzstyle{bl1} = [->, draw=black]   
\tikzstyle{bl2} = [->, draw=black!70,thick]   
\tikzstyle{bl3} = [draw=black,thick, dotted]   

\tikzstyle{br0} = [draw=brown, dashed]   
\tikzstyle{br1} = [->, draw=brown]   
\tikzstyle{br2} = [->, draw=brown,thick]   

\tikzstyle{red0} = [draw=red, thick, dashed]   
\tikzstyle{red1} = [draw=red]   
\tikzstyle{red2} = [draw=red,thick]   

\tikzstyle{gr0} = [draw=green, thick, dashed]   
\tikzstyle{gr1} = [draw=green]   
\tikzstyle{gr2} = [draw=green,thick]   
\tikzstyle{gr4} = [draw=green,semithick,rounded corners]   

\begin{tikzpicture}[scale=0.5][domain=0:8]
\draw (5,3) node[B-node,label=below:$v$] (v) {};
\draw (2,1) node[B-node,label=below:$u_d$] (u1) {};
\draw (2,3) node[B-node,label=above:$u_2$] (u2) {};
\draw (2,5) node[B-node,label=above:$u_1$] (u3) {};
\draw[bl2] (u1) -- (v) ;
\draw[bl2] (u2) -- (v) ;
\draw[bl2] (u3) -- (v) ;

\draw (2,1.5) node[test-node] {};
\draw (2,2) node[test-node] {};
\draw (2,2.5) node[test-node] {};

\draw[bl2] (0,0.5) -- (u1) ;
\draw[bl2] (0,1.5) -- (u1) ;

\draw[bl2] (0,2.5) -- (u2) ;
\draw[bl2] (0,3.5) -- (u2) ;

\draw[bl2] (0, 4.5) -- (u3) ;
\draw[bl2] (0, 5.5) -- (u3) ;
\draw (16,3) node[B-node,label=below:$v$] (v) {};
\draw (13,1) node[R-node,label=below:$f_v$] (fv) {};
\draw (13,3) node[R-node,label=below:$b_v$] (bv) {};
\draw (11,1) node[B-node,label=below:$u_d$] (u1) {};
\draw (11,3) node[B-node,label=above:$u_2$] (u2) {};
\draw (11,5) node[B-node,label=above:$u_1$] (u3) {};

\draw (9,0.5) node[R-node,label=left:$f_{u_d}$] (fu1) {};
\draw (9,1.5) node[R-node,label=left:$b_{u_d}$] (bu1) {};
\draw (9,2.5) node[R-node,label=left:$f_{u_2}$] (fu2) {};
\draw (9,3.5) node[R-node,label=left:$b_{u_2}$] (bu2) {};
\draw (9,4.5) node[R-node,label=left:$f_{u_1}$] (fu3) {};
\draw (9,5.5) node[R-node,label=left:$b_{u_1}$] (bu3) {};

\draw (11,1.5) node[test-node] {};
\draw (11,2) node[test-node] {};
\draw (11,2.5) node[test-node] {};

\draw[bl3] (bv) -- (v) ;
\draw[bl3] (fv) -- (v) ;

\draw[bl2] (u1) -- (bv) ;
\draw[bl2] (u2) -- (bv) ;
\draw[bl2] (u3) -- (bv) ;

\draw[bl3] (fu1) -- (u1) ;
\draw[bl3] (bu1) -- (u1) ;

\draw[bl3] (fu2) -- (u2) ;
\draw[bl3] (bu2) -- (u2) ;

\draw[bl3] (fu3) -- (u3) ;
\draw[bl3] (bu3) -- (u3) ;

\end{tikzpicture}

%% file: given-schedule.tex
\section{Computing Expected Number of Adopters}
\label{sec:comp-expect-numb}
Here we give an algorithm to compute $E(I_n)$, given a spreading strategy $\pi$ with thresholds given by vector $\boldsymbol{c}$ and initial probabilities of acceptance given by vector $\boldsymbol{p}$.
Let $Y_k$ be the number of $1$ decisions among vertices in $\{\pi(1),\pi(2),\ldots,\pi(k)\}$. We note that $I_k = 2Y_k - k$. 
Since $E(I_n) = \sum_{i \in \{1\ldots n\}} xPr(I_n=x)$, we are interested in computing $Pr(I_n=x), \fa x \in \{-n\ldots n\}$.
\givenschedule*
Let $A$ be a $n \times (2n+1)$ matrix where $A[k,x] = Pr(I_k=x), k \in \{1\ldots n\}, x \in \{-n\ldots n\}$. Let $v=\pi(k)$. The following recurrence might be used to arrive at a dynamic programming formulation:
\begin{align*}
  A[k,x] &\leftarrow Pr(X^k_v=1)A[k-1,x-1] + Pr(X^k_v=-1)A[k-1, x+1]
\end{align*}
However, one needs to be careful when computing $Pr(X^k_v=1)$ because it is dependent of $I_{k-1}$. Thus, in the correct recurrence we must have $Pr(X^k_v=1|I_{k-1}=x-1)$ and $Pr(X^k_v=-1|I_{k+1}=x+1)$ instead of $Pr(X^k_v=1)$ and $Pr(X^k_v=-1)$ respectively. Below we derive the dynamic program keeping this subtelty in mind.
Let $v=\pi(k+1)$. We have:
\begin{align*}
  Pr(I_{k+1}=x+1 | I_k = x) =& \begin{cases}
    p_{v}& \text{if } -c_v < x < c_v\\
    1 & \text{if }x \geq c_v\\
    0& \text{otherwise}    
  \end{cases}\\
  Pr(I_{k+1}=x-1|I_k = x) =& 1 - Pr(I_{k+1} = x+1|I_k = x)
\end{align*}
We have:
\begin{align*}
  Pr(I_{k+1} = x) =& Pr(I_{k+1}=x|I_k=x-1)Pr(I_k=x-1) \\&+ Pr(I_{k+1}=x|I_k=x+1)Pr(I_k=x+1) 
\end{align*}
The above relation suggests a dynamic program for computing $E(I_n)$. The matrix $A$ is initialized with $A[1,1] = p_{\pi(1)}, A[1,-1] = 1-A[1,1], A[1,0] = 0, A[k,x] = 0, \fa x > k, A[k,x] = 0, \fa x < -k$. 
When $|x| < n, k > 1$, then any $A[k,x]$ depends on $A[k-1, x+1]$ and $A[k-1, x+1]$ and we get the recurrence:
\begin{align*}
  A[k,x] \leftarrow& Pr(I_{k}=x|I_{k-1}=x-1)A[k-1,x-1] \\&+ Pr(I_k=x|I_{k-1}=x+1)A[k-1, x+1]
\end{align*}
From $A$, $E(I_n)$ can be computed as follows:
\begin{align*}
  E(I_n) = \sum_{i \in \{1\ldots n\}} xPr(I_n=x) = \sum_{i \in \{1\ldots n\}}iA[n,i]
\end{align*}

%% file: adaptive.tex
\section{Adaptive Marketing Strategy}
\label{sec:best_schedule_dp}
In this section we propose a dynamic program for computing best adaptive spreading strategy and thus, prove Theorem \ref{thm:best_schedule_dp}. Here we give dynamic program when there are two types of areas. This can be extended to any constant number of types. Let $B(n_1, n_2, k)$ be the expected number of areas that adopt the product for a best ordering where $n_1$ is number of areas of type $1$ and $n_2$ is the number of areas of type $2$ in the market $k$ is sum of decisions of vertices that have been scheduled so far. We note that deployment number $k$ is equal to difference of number of yes decisions and no decisions. Let thresholds and initial acceptance probabilities for vertices of type $i$ be $c_i$ and $p_i$. At any given time in the strategy, let $B_i$ be the best possible result if an area of type $i$ is scheduled next. Depending on value of $k$, we have the following cases (cases 2 and 4 will not occur if $c_1 = c_2$):
\begin{enumerate}
\item $n_1 = 0 \vee n_2 = 0$: If all areas are of the same type, then all spreading strategies are equivalent and we can choose any arbitraty spreading strategy for the remaining areas.
\item $c_1 \leq k < c_2$: In this case, areas of type $1$ will accept the idea \textbf{w.p.} $1$. Areas of type $2$ will accept the idea with probability $p_2$ and reject it with probability $1-p_2$.
  \begin{align*}
    B_1 =&1 + B(n_1-1, n_2, k+1)\\
    B_2 =&p_2 + p_2B(n_1 , n_2-1, k+1)+(1-p_2)B(n_1, n_2-1, k-1)\\
    B(n_1, n_2, k) =& \max\{B_1, B_2\}
  \end{align*}
\item $-c_1 < k < c_1$: In this case, both types of areas will decide to accept or reject the idea on basis of initial acceptance probabilities. Therefore:
  \begin{align*}
    B_1 =&p_1 + p_1B(n_1-1, n_2, k+1) + (1-p_1)B(n_1-1, n_2, k-1)\\
    B_2 =&p_2 + p_2B(n_1 , n_2-1, k+1)+(1-p_2)B(n_1, n_2-1, k-1)\\
    B(n_1, n_2, k) =&\max\{B_1, B_2\}
  \end{align*}
\item $-c_2 < k \leq -c_1$: In this case, areas of type $1$ will reject the idea with probability $1$ and areas of type $2$ will accept the idea with probability $p_2$.
  \begin{align*}
    B_1 =&B(n_1-1, n_2, k+1)\\
    B_2 =&p_2 + p_2B(n_1 , n_2-1, k+1)+(1-p_2)B(n_1, n_2-1, k-1)\\
    B(n_1, n_2, k) =& \max\{B_1, B_2\}
  \end{align*}
\item $k \leq -c_2$: In this case, both types of areas will reject the idea. Therefore:
  \begin{align*}
    B(n_1, n_2, k) = 0
  \end{align*}
\item $k \geq cc_2$: In this case, both types of areas will reject the idea. Therefore:
  \begin{align*}
    B(n_1, n_2, k) = n_1 + n_2
  \end{align*}
\end{enumerate}
This can easily be extended to any constant number of types. The time complexity with $t$ types is $O(n^{t+1})$.

%% file: missing.tex
\section{Missing Proofs}
\label{sec:miss}
\subsection{Proof of \lref{lem-triv}}
\label{sec:miss:lem1}
\begin{proof}
We prove this lemma by proving that:
\begin{align}
  Pr(I_{k+t} \geq x;\pi) &\geq Pr(I_{k+t} \geq x;\pi'), \fa t \in \{1\ldots n-k\}\label{eq:8}
\end{align} 
We note that the above implies $E(I_n; \pi) \geq E(I_n;\pi')$.
We prove that if $Pr(I_{k} \geq x;\pi) \geq Pr(I_{k} \geq x;\pi')$ then $Pr(I_{k+1} \geq x;\pi) \geq Pr(I_{k+1} \geq x;\pi')$ for all $x \in \mathbb{Z}$. This argument can be successively applied to prove (\ref{eq:8}).
Let $\pi(k+1) = v$. $X_v$ will be $1$ iff either $I_k \geq c_v$ and $v$ accepts idea based on threshold rule or $-c_v < I_k < c_v$ and $v$ decides to accept the idea based on initial acceptance probability $p_v$. Thus:
\begin{align}
  Pr(X_v = 1) =& Pr(I_k \geq c_v) + Pr(-c_v < I_k < c_v)p_v \nonumber
\end{align}
   Substituting $Pr(-c_v < I_k < c_v) = Pr(I_k \geq -c_v+1) - Pr(I_k \geq c_v)$, we have:
\begin{align}
  Pr(X_v = 1)=& Pr(I_k \geq c_v) + (Pr(I_k \geq -c_v+1) - Pr(I_k \geq c_v))p_v \nonumber
\end{align}
By rearranging the terms, we get:
\begin{align}
  Pr(X_v = 1)=& Pr(I_k \geq c_v)(1-p_v) + Pr(I_k \geq -c_v + 1)p_v\label{eq:6}
\end{align}
We are given that $Pr(I_{k} \geq x;\pi) \geq Pr(I_{k} \geq x;\pi'), \fa x \in \mathbb{Z}$. From this and from (\ref{eq:6}), we have, $Pr(X_v = 1;\pi) \geq Pr(X_v = 1;\pi')$. Thus, $Pr(I_{k+1} \geq x;\pi) \geq Pr(I_{k+1} \geq x;\pi'), \fa x \in \mathbb{Z}$.
\end{proof}